\documentclass[aip,jmp,numerical,reprint,onecolumn,eqsecnum,floats,showpacs]{revtex4-1}

\usepackage{amsmath}
\usepackage{amssymb}
\usepackage{mathrsfs}
\usepackage{amsthm}
\usepackage{textcomp}
\usepackage{xcolor}
\usepackage{graphicx}
\usepackage{enumerate}
\usepackage{mathtools}
\usepackage{hyperref}
\usepackage[title]{appendix}
\usepackage{subfigure}
\usepackage{bm}

\newtheorem{thm}{Theorem}[section]
\newtheorem{lem}[thm]{Lemma}
\newtheorem{prop}[thm]{Proposition}

\newtheorem{cor}[thm]{Corollary}

\newtheorem*{remark}{Remark}
\newtheorem*{remarks}{Remarks}

\newcommand{\A}{\mathcal{A}}
\newcommand{\eps}{\varepsilon}
\newcommand{\iPrime}{{\phantom{'}}}

\begin{document}

\title{Multistability of Phase-Locking in Equal-Frequency Kuramoto Models on Planar Graphs}

\author{Robin Delabays}
\affiliation{School of Engineering, University of Applied Sciences of Western Switzerland, CH-1950 Sion, Switzerland}
\affiliation{Section de Math\'ematiques, Universit\'e de Gen\`eve, CH-1211 Gen\`eve, Switzerland}
\author{Tommaso Coletta}
\affiliation{School of Engineering, University of Applied Sciences of Western Switzerland, CH-1950 Sion, Switzerland}
\author{Philippe Jacquod}
\affiliation{School of Engineering, University of Applied Sciences of Western Switzerland, CH-1950 Sion, Switzerland}


\date{8 September 2016, accepted 3 March 2017, published online 23 March 2017}

\begin{abstract}
 \it The number $\mathcal{N}$  of stable fixed points of locally coupled Kuramoto models depends on the topology of the network on which the model is defined. 
 It has been shown that cycles in meshed networks play a crucial role in determining $\mathcal{N}$, 
 because any two different stable fixed points differ by a collection of loop flows on those cycles. 
 Since the number of different loop flows increases with the length of the cycle that carries them, one expects $\mathcal{N}$ to be larger in meshed networks with longer cycles.
 Simultaneously, the existence of more cycles in a network means more freedom to choose the location of loop flows differentiating between two stable fixed points. 
 Therefore, $\mathcal{N}$ should also be larger in networks with more cycles. 
 We derive an algebraic upper bound for the number of stable fixed points of the Kuramoto model with identical frequencies, 
 under the assumption that angle differences between connected nodes do not exceed $\pi/2$.
 We obtain $\mathcal{N}\leq\prod_{k=1}^c\left[2\cdot{\rm Int}(n_k/4)+1\right]$, which depends both on the number $c$ of cycles and on the spectrum of their lengths $\{n_k\}$.
 We further identify network topologies carrying stable fixed points with angle differences larger than $\pi/2$, 
 which leads us to conjecture an upper bound for the number of stable fixed points for Kuramoto models on any planar network. 
 Compared to earlier approaches that give exponential upper bounds in the total number of vertices, 
 our bounds are much lower and therefore much closer to the true number of stable fixed points. 
\end{abstract}

\maketitle

\section{Introduction}

To analyze and describe phenomena of "temporal organization of matter", 
Kuramoto proposed a model of coupled oscillators defined by the following set of nonlinear differential equations~\cite{Kur75abbrv,Kur84}
\begin{eqnarray}\label{eq:Kuramoto}
\dot{\theta}_i = P_i - \sum_{j=1}^n K_{ij} \sin(\theta_i - \theta_j) \, , \quad i=1,\ldots n\, .
\end{eqnarray}
These coupled differential equations govern the dynamics of the angular coordinates $\theta_i$ of $n$ one-dimensional oscillators with natural frequencies $P_i$, 
interacting with one another via a coupling that is odd and periodic in angle differences. 
Since Kuramoto's original paper,~\cite{Kur75abbrv} the model has established itself as a cornerstone in the theory of synchrony and synchronization phenomena in coupled dynamical systems, 
in great part thanks to its simplicity, which allows for analytical treatments but still captures many key ingredients of these complex problems.~\cite{Str00,Ace05,Dor14}

For all-to-all coupling and for $n \to \infty$, a mean-field treatment based on the continuous limit approximation facilitates the analytical treatment of the Kuramoto model.
In particular it was observed early on that for all-to-all constant coupling, $K_{ij} \equiv K/n$, 
a stable synchronous state $\{ \theta_i^{(0)} \}$ emerges for $K > K_c$ with $\dot\theta_i^{(0)}-\dot\theta_j^{(0)}=0$, for a finite fraction of pairs of oscillators $(i,j)$. 
The critical coupling strength $K_c$ depends on the distribution $g(P)$ of natural frequencies $P_i$.~\cite{Kur84} 
Furthermore, full phase-locking with $\dot\theta_i^{(0)}-\dot\theta_j^{(0)}=0$, for all $i,j$ is possible for large enough $K$ provided $g(P)$ has compact support.~\cite{Erm85,Hem93}

It is the rule, however, rather than the exception that physical systems of coupled oscillators have short-range instead of all-to-all coupling. 
Short-range coupling renders the problem significantly harder, because the mean-field approach breaks down, but simultaneously more interesting, 
most notably because there may be more than one stable fixed point solution, contrarily to the model with all-to-all coupling.~\cite{Aey04,Mir05} 
As a matter of fact, determining the number of linearly stable fixed point solutions for locally coupled Kuramoto models is an unsolved mathematical problem. 
Fixed point solutions to Eq.~(\ref{eq:Kuramoto}) are determined by
\begin{eqnarray}\label{eq:Ksteady}
P_i = \sum_{j=1}^n K_{ij} \sin(\theta_i - \theta_j) \, , 
\end{eqnarray}
i.e. by a set of $n$ nonlinear transcendental equations. 
Eqs.~(\ref{eq:Kuramoto}) and (\ref{eq:Ksteady}) are invariant under $\theta_i \rightarrow \theta_i + \varphi_0$, 
so that it is trivial to construct infinitely many fixed points starting from a given fixed point. 
Solutions are however usually counted modulo constant angle shifts. 
Upper bounds for the number ${\cal N}$ of fixed point solutions of Eq.~(\ref{eq:Kuramoto}) exist, which are based on algebraic geometrization methods.~\cite{Meh15,Che16} 
Algebraic geometrization consists in rewriting Eq.~(\ref{eq:Kuramoto}) as polynomials in trigonometric functions of angles. 
The method provides a list of candidate fixed points of Eq.~\eqref{eq:Kuramoto} whose stability 
needs to be confirmed or discarded numerically. 
This method guarantees that all fixed points are found but leads to a bound on the number of stable fixed points which is exponential in the number $n$ of nodes.
Exponential bounds are easy to understand qualitatively since new fixed point solutions to Eq.~(\ref{eq:Ksteady}) can be obtained in principle from any known
fixed point solution via $\theta_i^{(0)}  - \theta_j^{(0)}  \rightarrow \pi-(\theta_i^{(0)}  - \theta_j^{(0)} )$ for any edge $\langle ij\rangle$.~\cite{Rog04} 
The latter substitution leads to an exponential number $\propto 2^{n}$ of fixed point solutions. 
It is expectable that much better bounds exist. 
As a matter of fact, the vast majority of fixed point solutions obtained with the above substitution are unstable and thus of lesser physical significance.

Another way to obtain a bound on ${\cal N}$ is to rely on brute-force numerics and calculate how many different fixed points one converges to, 
starting from a large number of initial conditions $\{\theta_i(t=0) \}$. 
This method is however limited, first, by the small volume of the basin of attraction of some fixed points,~\cite{Wil06,Men14} which require a finite resolution in angle space and second, 
by the exponentially large number $(2 \pi/\Delta)^{n}$ of initial conditions needed to have a given angular resolution of $\Delta$ in angle space. 
Moderate resolution and number of oscillators, $\Delta = 0.5$ and $n=20$, already require an impractically large number $\simeq 10^{21}$ of initial conditions. 
Clearly, a novel approach is needed to estimate ${\cal N}$.

Below we present such a novel approach and derive an upper bound on the number of stable fixed point solutions to Eq.~\eqref{eq:Kuramoto} 
defined on a planar network with equal natural frequencies, $P_i \equiv \Omega$. 
That bound is generally valid, under the only assumption that angle differences in the stable fixed points do not exceed $\pi/2$. 
We further identify networks where stable fixed points exist which violate that latter condition and conjecture an upper bound on ${\cal N}$ valid in that case. 
Our method is based on the theorem of Refs.~\onlinecite{Dor13,Del16} which shows that, if there is more than one fixed point solution to Eq.~(\ref{eq:Kuramoto}), 
these fixed point solutions differ by a collection of loop flows (to be defined below) around the cycles of the network. 
Inspired by this result, we explore the possibility to generate such loop flows around each and every cycle in the network,
and investigate how many of the fixed points obtained in this way remain linearly stable.

In this manuscript, we amplify on our earlier work~\cite{Del16} which gave an algebraic upper bound ${\cal N} \le 2 \cdot {\rm Int}[n/4]+1$ 
for the number of stable fixed points of Eq.~(\ref{eq:Kuramoto}) on a single-cycle network of length $n$, where ${\rm Int}[\cdot]$ denotes the integer part (floor) function. 
Ref.~\onlinecite{Del16} follows a rather long series of investigations which we briefly summarize here. 
The question of the number of stable fixed point solutions to systems of coupled oscillators dates back at least to 
Korsak~\cite{Kor72} and Tavora and Smith~\cite{Tav72} in the context of electrical power flow problems, 
which are related to the Kuramoto model in the lossless line approximation (see e.g. Refs.~\onlinecite{Del16,Ber00}).
Korsak explicitly calculated different, linearly stable fixed points for a simple network, 
which differ from one another by a circulating electrical current around a cycle formed by the network. 
This was followed by a number of investigations, most of them numerical or on small networks in the context of electrical power grids, 
see e.g. Refs.~\onlinecite{Ska80,Ara81,Bai82,Tam83,Klos91} for a necessarily incomplete, but representative list of the large literature in that direction.
Bounds for the number of fixed point solutions in the spirit of the substitution argument given after Eq.~(\ref{eq:Ksteady}) where constructed in 
Refs.~\onlinecite{Bai82,Ngu14,Meh15}, which gave exponential upper bounds for the number of stable fixed points of Kuramoto/power flow models. 
Ref.~\onlinecite{Meh15} observed numerically, however, that the number of stable fixed point solutions is much smaller than this bound.

Tavora and Smith~\cite{Tav72} realized that circulating loop flows such as those discussed by Korsak~\cite{Kor72} can only take discrete values, because angles are defined modulo $2 \pi$. 
As a matter of fact, single-valuedness of angle coordinates requires that summing over angle differences around any cycle in the network on which Eqs.~(\ref{eq:Kuramoto})
and (\ref{eq:Ksteady}) are defined must give an integer multiple of $2 \pi$ in order to get back to the initial angle, modulo $2 \pi$.
This defines topological winding numbers $q_k$,
\begin{eqnarray}\label{windingNumber}
 q_k &=& \left(2\pi\right)^{-1}\sum_{\ell=1}^{n_k} \Delta_{\ell+1,\ell}\, \in \, \mathbb{Z}\, ,
\end{eqnarray}
which are integers. The the sum runs over all $n_k$ nodes (locating individual oscillators) around any (the $k^{\rm th}$) cycle in the network on which the model is defined and node 
indices on a cycle satisfy $n_k+1 \rightarrow 1$. 
We introduce $\Delta_{ij}\coloneqq\theta_i-\theta_j$ taken modulo $2\pi$ in the interval $(-\pi,\pi]$. 
We will alternatively use the notation $\Delta_e=\Delta_{ij}$, where $e=\langle ij\rangle$ is a single edge index. 
The topological meaning of $q_k$ is obvious, as it counts the number of times the oscillators angles wind around the origin in the complex plane as one goes around the $k^{\rm th}$ cycle. 
This latter observation eventually led to the characterization of circulating loop flows with discrete topological winding numbers.~\cite{Jan03,Erm85,Del16,Col16b}
Using winding numbers, Rogge and Aeyels~\cite{Rog04} obtained an algebraic upper bound $\mathcal{N} \le 2\cdot {\rm Int}[n/4]+1$ 
for the number of stable fixed point solutions with any angle difference in a Kuramoto model on a $n$-node ring with unidirectional nearest-neighbor couplings.
The same upper bound has been calculated by Ochab and G\'ora~\cite{Och10} in a nonoriented $n$-node Kuramoto ring with 
nearest-neighbor coupling, under the condition that all angle differences are smaller than $\pi/2$. 
This upper bound is reached when the coupling strength goes to infinity, equivalently corresponding to $P_i=0$, for all $i$ in 
Eqs.~(\ref{eq:Kuramoto}) and (\ref{eq:Ksteady}), i.e. to identical oscillators.
Refs.~\onlinecite{Rog04,Och10} on single-cycle networks were complemented by our recent work,~\cite{Del16} which showed that the bound $\mathcal{N} \le 2\cdot {\rm Int}[n/4]+1$ 
is still valid for nonoriented coupling, even when one angle difference exceeds $\pi/2$, and that no stable fixed point exists
with more than one angle difference exceeding $\pi/2$ for single-cycle networks. 
These results are summarized in Table~\ref{tab:bounds}.

\begin{table}
\begin{tabular}{|c|c|c|}
 \hline
 Reference & Number of fixed points bounded by & Conditions on graph \\
 \hline
 \hline
 \onlinecite{Rog04} & $2\cdot{\rm Int}(n/4)+1$ & Single-cycle, unidirectional coupling \\
 \hline
 \onlinecite{Och10} & $2\cdot{\rm Int}(n/4)+1$ & Single-cycle, $\Delta_{ij}<\pi/2$, $\forall \langle ij\rangle$ \\
 \hline
 \onlinecite{Del16} & $2\cdot{\rm Int}(n/4)+1$ & Single-cycle \\
 \hline
 \onlinecite{Meh15}, \onlinecite{Che16} & $\propto 2^n$ & Any network \\
 \hline
 \hline
 Present paper & $\prod_{k=1}^c\left[2\cdot{\rm Int}(n_k/4)+1\right]$ & Planar, no single-shared edge, $P_i=P_0$, $\forall i$ \\
 \hline
\end{tabular}
\caption{Existing bounds on the number of stable fixed points of the Kuramoto model.}
\label{tab:bounds}
\end{table}

We finally note that Ref.~\onlinecite{Til11}  investigated different stable fixed point solutions, which were later classified according
to two integers, $q$ (the winding number mentioned above) and $\ell$ (the number of angle differences larger than $\pi/2$) in Ref.~\onlinecite{Roy12}. 
Ref.~\onlinecite{Del16} put this classification into perspective as it showed that, based on a theorem due to Taylor~\cite{Tay12} (see also Ref.~\onlinecite{Do12}), 
only $\ell=0,1$ is possible on a single-cycle network. 
In a somewhat different but related direction of investigation, Wiley et al.~\cite{Wil06} investigated the size of basins of attraction for 
synchronous fixed point solutions with different $q$ in a cycle network of identical oscillators and proposed that they are smaller at higher $q$.

We think the present manuscript makes a significant step forward in generalizing the bound found in and the method of Ref.~\onlinecite{Del16} to multi-cycle, planar networks. 
The problem here is threefold. 
First, one must account for an effective coupling between neighboring cycles with common edges. 
Along the latter, loop flows are superimposed, which couples two conditions expressed in Eq.~\eqref{windingNumber} for two different cycles. 
Second, this superposition leads to increased or decreased angle differences, which sometimes jeopardizes the stability of the resulting fixed point solutions. 
Third, in meshed networks, one has to revisit the one-to-one relation between number of stable fixed points and number of possible winding numbers on which the strategy of 
Ref.~\onlinecite{Del16} is based. 
Below, we present a method that overcomes these obstacles, which allows us to generalize the bound found in Refs.~\onlinecite{Rog04,Och10,Del16}.

The paper is organized as follows: 
In Section~\ref{sec:model} we formally present the model we investigate. 
We derive our upper bound on the number of stable fixed points of Eq.~\eqref{eq:Ksteady} with $P_i\equiv\Omega$ in Section~\ref{sec:bound}. 
In Section~\ref{sec:anglediff}, we give two examples of networks where fixed points with some angle differences larger than $\pi/2$ are stable. 
A conclusive discussion is given in Section~\ref{sec:conclusion}.

\section{The model}\label{sec:model}
Let $G$ be a connected planar graph with $n$ vertices and $m$ edges, and let $P_i\in\mathbb{R}$, for $i=1,...,n$, be the natural frequencies of one-dimensional oscillators located at each vertex. 
A variable angle $\theta_i\in\mathbb{R}$, $i=1,...,n$, is associated to each vertex and a coupling
constant $K\geq0$ to every edge $\langle ij\rangle$. 
We are interested in the stationary points of the Kuramoto model defined on $G$, whose dynamics is defined by 
\begin{align}\label{eq:kuramoto_general}
 \dot{\theta_i}&=P_i-K\sum_{j\sim i}\sin\left(\theta_i-\theta_j\right)\, ,
\end{align}
where $j\sim i$ indicates that the sum is taken over the neighboring vertices $j$ of $i$. 
From now on, we assume identical frequencies, $P_i=\Omega$, for all $i$, and consider the angles in a rotating frame, $\theta_i\to\Omega t+\theta_i$. 
Eq.~\eqref{eq:kuramoto_general} then reads
\begin{align}\label{eq:kuramoto_red}
 \dot{\theta_i}&=-K \sum_{j\sim i} \,\sin\left(\theta_i-\theta_j\right)\, . 
\end{align}

The linear stability of a fixed point $\{\theta_i^{(0)}\}$ of Eq.~\eqref{eq:kuramoto_red} is assessed by considering small perturbations about this fixed point, 
$\theta_i^{(0)}\to\theta_i^{(0)}+\delta\theta_i$ and linearizing the time evolution of $\delta\theta_i$. 
One obtains 
\begin{align*}
 \delta\dot{\theta_i}&=-\sum_{j\sim i}K\cos(\theta_i^{(0)}-\theta_j^{(0)})(\delta\theta_i-\delta\theta_j), \quad i=1,...,n\, .
\end{align*}
The stability of the fixed point is then determined by the \emph{Lyapunov exponents} which are the eigenvalues of the \emph{stability matrix} $M(\{\theta_i^{(0)}\})$, 
\begin{align}\label{eq:stab_matrix}
 M_{ij}&\coloneqq\left\{
 \begin{array}{ll}
  K\cos(\theta_i^{(0)}-\theta_j^{(0)})\, , & \text{if }i\sim j\, ,\\
  \displaystyle -\sum_{k\sim i}K\cos(\theta_i^{(0)}-\theta_k^{(0)})\, , & \text{if }i=j\, ,\\
  0\, , & \text{otherwise.}
 \end{array}
 \right.
\end{align}
The matrix $M$ is real symmetric, its spectrum is then real. 
Obviously, $\sum_iM_{ij}=\sum_jM_{ij}=0$, implying that one eigenvalue is always zero, $\lambda_1=0$. 
A fixed point is then stable if and only if the stability matrix $M$ is negative semidefinite, i.e. when
its largest nonvanishing eigenvalue is negative, $\lambda_2<0$ .

Any fixed point of Eq.~\eqref{eq:kuramoto_red} satisfies
\begin{align}\label{eq:kirchhoff}
 \sum_{j\sim i}\sin(\Delta_{ij})&=0\, ,\quad i=1,...,n\, ,
\end{align}
where $\Delta_{ij}$ is defined after Eq.~\eqref{windingNumber}.
An obvious fixed point is given by $\Delta_{ij}=0$ for all $\langle ij\rangle$, i.e. all the angles are equal, but  
other fixed points may exist.~\cite{Kor72,Jan03,Rog04,Och10}
A simple example is provided by the single-cycle networks considered in 
Refs.~\onlinecite{Kor72,Jan03,Rog04,Och10,Del16}, where Eq.~(\ref{eq:kirchhoff}) can be
satisfied if $\sin(\Delta_{i,i-1})=\sin(\Delta_{i+1,i})$. When $\sin(\Delta_{i+1,i}) \ne 0$, a finite loop
flow of magnitude $K \eps = K \sin(\Delta_{i+1,i})$ circulates around the cycle. 
Because winding numbers are integers [Eq.~(\ref{windingNumber})], $\eps$ and $\Delta_{i+1,i}$ can take only discrete values. 
Ref.~\onlinecite{Col16b} made a connection between that discreteness and quantization of circulation around a superfluid vortex. 
Accordingly loop flows are alternatively referred to as {\it vortex flows}.

Vortex flows are not limited to single-cycle networks. 
According to Refs.~\onlinecite{Dor13,Del16}, two fixed points differ by a collection of loop flows on any of the different cycles of $G$.
For a vertex $i$ on a given cycle of $G$, such a flow increases the interaction strength between $i$ and one of its neighbors on the cycle, $\sin(\Delta_{ij})\to\sin(\Delta_{ij})+\delta$, 
but decreases the interation strength with its other neighbor in the cycle by the same amount, $\sin(\Delta_{ik})\to\sin(\Delta_{ik})-\delta$. 
Eq.~\eqref{eq:kirchhoff} is then still satisfied.

For any connected graph $G$ with $n$ vertices and $m$ edges, we can define a set of \emph{fundamental cycles}. 
Let $T\subset G$ be a spanning tree of $G$ and $r$ an arbitrary vertex of $G$ which we call the \emph{root}. 
Each edge of $W\coloneqq G\setminus T$ closes a cycle of $G$. 
The fundamental cycle associated to an edge $\langle ij\rangle\in W$ is the set of edges obtained by 
concatenation of the unique path on $T$ from $r$ to $i$, the edge $\langle ij\rangle$ and the unique path on $T$ from $j$ to $r$, with the prescription that,
if the resulting path goes twice through an edge, then we remove the latter from the path. 
This gives us $c=m-n+1$ fundamental cycles which form together a cycle basis. 
Any cycle of $G$ is a linear combination of fundamental cycles (see Ref.~\onlinecite{Big93} for more details). 
If the graph $G$ is planar, a wise choice of the spanning tree $T$ gives a set of fundamental cycles defined by the edges surrounding each face of the embedding of $G$ in $\mathbb{R}^2$. 
From now on, we focus on connected planar graphs and choose this fundamental cycle basis. 
We do not consider tree-like parts on $G$, since 
angle differences on such parts are uniquely defined and have no influence on the dynamics of the 
rest of the network. Tree-like parts can therefore be absorbed in the node that connects
them to the meshed part of $G$, which is why we neglect them in the following.

Let $\mathcal{L}_1,...,\mathcal{L}_c$ be the sets of edges composing each fundamental cycle. 
For $k=1,...,c$, we define $m_k$ the number of edges exclusively on cycle $k$ and for $k\neq\ell$, $m_{k\ell}$ is the number of edges common to the two cycles $k$ and $\ell$. 
Because we restrict ourselves to planar graphs, no edge is common to more than two cycles.
We then define $n_k\coloneqq m_k+\sum_{\ell\neq k}m_{k\ell}$ the number of elements in $\mathcal{L}_k$. 
We define $K\eps_k$, the value of the vortex flow on cycle $\mathcal{L}_k$, for $k=1,...,c$, where $\eps_k\in\mathbb{R}$ is the \emph{vortex flow parameter}. 
Any vortex flow can be uniquely written as a linear combination of vortex flows on fundamental cycles.

We arbitrarily fix an orientation to each edge and each cycle independently. 
We can now define the \emph{edge-cycle incidence matrix}, $S\in\mathbb{R}^{m\times c}$ as follows
\begin{align*}
 S_{ek}\coloneqq\left\{
 \begin{array}{ll}
  +1 & \text{if }e\in\mathcal{L}_k\text{ with the same orientation;}\\
  -1 & \text{if }e\in\mathcal{L}_k\text{ with opposite orientation;}\\
  0 & \text{if }e\notin\mathcal{L}_k\, ,
 \end{array}
 \right.
\end{align*}
where $e$ is an edge index and $k$ a cycle index. 
We define the interaction strength on edge $\langle ij\rangle$ as $P_{ij}\coloneqq K\sin(\Delta_{ij})$.
As the columns of $S$ are the edge vectors of the fundamental cycles, the interaction strength on any edge $e=\langle ij\rangle$ can be written as 
\begin{align}
 P_{ij}&=K\sum_{k\colon e\in\mathcal{L}_k} S_{ek}\eps_{k}\, ,
\end{align}
where the sum is taken over all cycles containing edge $e$. 
The magnitude of the interaction strength on each edge is bounded by the coupling on that edge, therefore, for every edge $e=\langle ij\rangle$,
\begin{align*}
 -K\leq P_{ij}\leq K\implies -1\leq\sum_{k\colon e\in\mathcal{L}_k}S_{ek}\eps_k\leq1\, ,
\end{align*}
which defines a parameter domain $\mathcal{D}_G\subset\mathbb{R}^c$ of possible values for the vector of vortex flows $\vec{\eps}\coloneqq(\eps_1,...,\eps_c)$.

From now on, we consider stable fixed points of Eq.~\eqref{eq:kuramoto_red} with all angle differences in $[-\pi/2,\pi/2]$. 
Let $\vec{\eps}\in\mathcal{D}_G$ be the vortex flow vector of a fixed point of Eq.~\eqref{eq:kuramoto_red}. 
The angle difference on edge $e=\langle ij\rangle$ is $\Delta_{ij}\in [-\pi/2,\pi/2]$, and is given by
\begin{align}\label{eq:angleDiff}
 \Delta_{ij}(\vec{\eps})&=\arcsin\left(P_{ij}(\vec{\eps})/K\right)=\arcsin\left(\sum_{k\colon e\in\mathcal{L}_k}S_{ek}\eps_k\right)\, .
\end{align}
The angle variables $\theta_i$ are single-valued, therefore the sum of angle differences (taken modulo $2\pi$) around any cycle must be an integer multiple of $2\pi$, 
\begin{align}\label{eq:quantization}
 \sum_{e\in\mathcal{L}_k}\Delta_e(\vec{\eps})&=2\pi q_k\, ,\quad k=1,...,c\, ,
\end{align}
which defines the \emph{winding number} $q_k\in\mathbb{Z}$ on cycle $\mathcal{L}_k$.
Eq.~(\ref{eq:quantization}) is equivalent to Eq.~(\ref{windingNumber}), with the sum being taken over edges instead of vertices.  
It states in particular that to any fixed point of Eq.~\eqref{eq:kuramoto_red} we can associate a 
point $\vec{\eps}\in\mathcal{D}_G$, but that not any $\vec{\eps}$ gives a fixed point of Eq.~\eqref{eq:kuramoto_red}.

\begin{remark}
 If a fixed point has all angle differences in $[-\pi/2,\pi/2]$, then its stability matrix is diagonal dominant. 
 According to Gershgorin's Circle Theorem,~\cite{Hor86} the stability matrix is then negative semidefinite and the fixed point is stable. 
 This was previously noted in Ref.~\onlinecite{Tav72}.
\end{remark}

Define the \emph{quantization function} 
\begin{align*}
 \vec{\A}&=(\A_1,...,\A_c)\colon\mathcal{D}_G\to\mathbb{R}^c\, ,
\end{align*}
whose components are the sums of angle differences around the corresponding fundamental cycles, 
\begin{align*}
 \A_k(\vec{\eps})&\coloneqq\sum_{e\in \mathcal{L}_k}\Delta_e(\vec{\eps})
 =\sum_{e\in \mathcal{L}_k}\arcsin\left(\sum_{\ell\colon e\in\mathcal{L}_\ell}S_{e\ell}\eps_\ell\right)\, ,
 \quad k=1,...,c\, .
\end{align*}

\begin{remark}
 We call $\vec{\A}$ quantization function because Eq.~\eqref{eq:quantization} leads to $\A_k=2\pi q_k$, which is analogous to quantization conditions in quantum mechanics 
 such as quantization of circulation around vortices in superfluids or of fluxes through superconducting rings. 
 Note however that there is no quantum mechanics in the Kuramoto model. 
\end{remark}

Define next the level sets $\mathcal{Q}_k^{q}\coloneqq\left\{\vec{\eps}\in\mathcal{D}_G\colon\A_k(\vec{\eps})=2\pi q\right\}$, for $k=1,...,c$ and $q\in\mathbb{Z}$. 
It is easy to check that for any $k$, there are at most
\[
 \mathcal{N}_k=2\cdot{\rm Int}[n_k/4]+1
\]
distinct level sets for $\A_k$ in the domain $\mathcal{D}_G$.~\cite{Del16} 
Define finally the intersection of level sets
\begin{align}\label{eq:qq}
 \mathcal{Q}(\vec{q})\coloneqq\bigcap_{k=1}^c\mathcal{Q}_k^{q_k}\, ,\quad \vec{q}=(q_1,...,q_c)\in\mathbb{Z}^c\, .
\end{align}

These definitions bring us to the following proposition, which will be useful in Sec.~\ref{sec:bound}. 

\begin{prop}\label{prop:1to1}
 Stable fixed point solutions of Eq.~\eqref{eq:kuramoto_red} with winding vectors $\vec{q}=(q_1,...,q_c)$ and angle differences in $[-\pi/2,\pi/2]$ 
 are in a one-to-one correspondence with vortex flow vectors $\vec{\eps}\in\mathcal{Q}(\vec{q})$. 
\end{prop}

\begin{proof}
 Eq.~\eqref{eq:angleDiff} associates a fixed point of Eq.~\eqref{eq:kuramoto_red} to each $\vec{\eps}\in\mathcal{Q}(\vec{q})$. 
 By assumption, this fixed point has all angle differences in $[-\pi/2,\pi/2]$ and winding vector $\vec{q}$ with components given by Eq.~\eqref{eq:quantization}. 
 Any fixed point of Eq.~\eqref{eq:kuramoto_red} is a collection of vortex flows, i.e. a point $\vec{\eps}\in\mathcal{D}_G$. 
 By definition again, this point belongs to $\mathcal{Q}(\vec{q})$. 
 Assume that $\mathcal{Q}(\vec{q})$ is not empty and let $\vec{\eta},\vec{\mu}\in\mathcal{Q}(\vec{q})$, 
 such that $\Delta_e(\vec{\eta})=\Delta_e(\vec{\mu})$, for all edges $e$. 
 This means that for each edge $e$, 
 \begin{align*}
  \sum_{k=1}^cS_{ek}\eta_k&=\sum_{k=1}^cS_{ek}\mu_k\, ,
 \end{align*}
 and in matricial form
 \begin{align*}
  S\cdot\left(\vec{\eta}-\vec{\mu}\right)&=0\, .
 \end{align*}
 As the columns of $S$ correspond to a cycle basis of $G$, they are linearly independent, implying $\vec{\eta}=\vec{\mu}$. 
 This concludes the proof of the one-to-one correspondence. 
\end{proof}

\section{Bound on the number of stable fixed points}\label{sec:bound}
In this section, we show that for a given winding vector, $\vec{q}\in\mathbb{Z}^c$, the corresponding level sets $\mathcal{Q}_k^{q_k}$, $k=1,...,c$ intersect at most at one point. 
We then derive some bounds on the number of stable fixed points of Eq.~\eqref{eq:kuramoto_red} with all angle differences less than $\pi/2$. 
The following theorem is key to obtain this bound, which we later derive in Corollary~\ref{cor:main}. 

\begin{thm}\label{thm:main_thm}
 Let $G$ be a planar graph and $\vec{\A}\colon\mathcal{D}_G\to\mathbb{R}^c$ the quantization function associated with its fundamental cycle basis defined on the faces of the embedding of $G$
 in $\mathbb{R}^2$. 
 For a given winding vector $\vec{q}\in\mathbb{Z}^c$, the intersection, Eq.~\eqref{eq:qq}, of its level sets is either a single point or empty.
\end{thm}

\begin{proof}
 For a planar network composed of $c$ fundamental cycles, Eq.~\eqref{eq:quantization} reads
 \begin{align}\label{eq:arcsin_gen}
  \A_k(\vec{\eps})\coloneqq m_k\arcsin(\eps_k)+\sum_{i\neq k}m_{ki}\arcsin(\eps_k-\eps_i)=2\pi q_k\, , \quad k=1,...,c\, .
 \end{align}
 For each $k$, Eq.~\eqref{eq:arcsin_gen} defines a level set, $\mathcal{Q}_k^{q_k}\subset\mathcal{D}_G$ of possible values for $\vec{\eps}$. 
 Assume that these level sets intersect in two distinct points of $\mathcal{D}_G$, $\vec{\eta}$ and $\vec{\mu}$. 
 Let $\vec{\xi}\coloneqq\vec{\mu}-\vec{\eta}$ be the difference between these two points. 
 Because $\vec{\eta}$ and $\vec{\mu}$ are assumed dinstict, $\vec{\xi}$ has at least one non-zero component. 
 Without loss of generality, we assume this component to be positive. 
 We then order the cycles such that $\xi_1=\max_\ell\{\xi_\ell\}>0$. 
 We now consider the directional derivative of $\A_k$ for $k=1,...,c$ in the direction of $\vec{\xi}$,
 \begin{align}\label{eq:dirderk}
  \nabla\A_k(\vec{\eps})\cdot\vec{\xi}&=\frac{m_k\xi_k}{\sqrt{1^\iPrime-\eps_k^2}}+\sum_{i\neq k}\frac{m_{ki}(\xi_k-\xi_i)}{\sqrt{1-(\eps_k^\iPrime-\eps_i^\iPrime)^2}}\, .
 \end{align}
 At each end of the segment defined by $\vec{\eta}+\alpha\vec{\xi}$, $\alpha\in[0,1]$, the function $\A_k$ takes the same value. 
 Thus the directional derivative either is constant and equal to zero or changes sign for some $\alpha\in(0,1)$. 
 Consider Eq.~\eqref{eq:dirderk} with $k=1$,
 \begin{align}\label{eq:dirder}
  \nabla\A_1(\vec{\eps})\cdot\vec{\xi}&=\frac{m_1\xi_1}{\sqrt{1^\iPrime-\eps_1^2}}+\sum_{i\neq 1}\frac{m_{1i}(\xi_1-\xi_i)}{\sqrt{1-(\eps_1^\iPrime-\eps_i^\iPrime)^2}}\, .
 \end{align}
 
 First if $m_1>0$, then the sum of terms on the right-hand-side is strictly positive. 
 Thus the directional derivative of $\A_1$ is never zero for $\alpha\in(0,1)$.
 
 Second, if $m_1=0$ and there exist a cycle $\ell$ such that $m_{1\ell}>0$ and $\xi_\ell<\xi_1$, then the right-hand-side of Eq.~\eqref{eq:dirder} is positive. 
 
 Third, if $m_1=0$ and for all $\ell$ such that $m_{1\ell}>0$, $\xi_1=\xi_\ell$, then either $\xi_k=\xi_1$ for all $k=1,...,c$, in which case there exist $k_0$ such that $m_{k_0}>0$, 
 or there exist $k_0$ and $k_1$ such that $\xi_1=\xi_{k_0}>\xi_{k_1}$ and $m_{k_0k_1}> 0$. 
 In both cases, considering $\A_{k_0}$ instead of $\A_1$ in Eq.~\eqref{eq:dirder}, we obtain that the right-hand-side is strictly positive. 
 
 Therefore, the directional derivative on the left-hand-side of Eq.~\eqref{eq:dirder} does not vanish along $\vec{\eps}=\vec{\eta}+\alpha\vec{\xi}$ for any $\alpha$. 
 This contradicts the assumption that $\vec{\eta}$ and $\vec{\mu}$ are distinct points of the intersection of level sets $\mathcal{Q}(\vec{q})$ and concludes the proof. 
\end{proof}

Theorem~\ref{thm:main_thm} implies that the set of stable fixed points of Eq.~\eqref{eq:kuramoto_red} with angle differences in $[-\pi/2,\pi/2]$ injects in the set of possible winding vectors. 
In other words, two distinct fixed points with angle differences in $[-\pi/2,\pi/2]$ have distinct winding vectors. 
The number of stable fixed points of Eq.~\eqref{eq:kuramoto_red} is then bounded by the number of possible winding vectors. 

\begin{cor}\label{cor:main}
 For any planar graph $G$, the number $\mathcal{N}^*$ of stable fixed points of Eq.~\eqref{eq:kuramoto_red} with angle differences in $[-\pi/2,\pi/2]$ is bounded from above as
 \begin{align}\label{eq:bound_restr}
  \mathcal{N}^*&\leq\prod_{k=1}^c\left[2\cdot{\rm Int}\left(n_k/4\right)+1\right]\, .
 \end{align}
\end{cor}

\begin{proof}
 For all $k=1,...,c$, the function $\A_k$ takes values in $[-n_k\pi/2,n_k\pi/2]$ and thus $q_k\in\{-{\rm Int}[n_k/4],...,{\rm Int}[n_k/4]\}$. 
 There are then at most $\prod_{k=1}^c\left[2\cdot{\rm Int}\left(n_k/4\right)+1\right]$ possible winding vectors $\vec{q}$. 
 According to Theorem~\ref{thm:main_thm}, for each of these winding vectors, the corresponding level sets $\mathcal{Q}_k^{q_k}$ intersect at most at a single point. 
 The one-to-one correspondence between fixed points of Eq.~\eqref{eq:kuramoto_red} and intersections of these level sets expressed in Proposition~\ref{prop:1to1} concludes the proof. 
\end{proof}

\begin{remarks}
 \begin{enumerate}[(i)]
  \item Note that in Corollary~\ref{cor:main} we restrict ourselves to stable fixed points with all angle differences in $[-\pi/2,\pi/2]$ whose number $\mathcal{N}^*$ is, 
  in general, lower than the total number of stable fixed points $\mathcal{N}$. 
  \item The bound on $\mathcal{N}^*$ is algebraic in the length of the cycles. 
  This is a significant improvement on the previously known bounds which are exponential in the number of vertices. 
  Nevertheless, the bound of Corollary~\ref{cor:main} is not yet tight, because some choices of $\vec{q}=(q_1,...,q_c)$ are not realizable. 
 \end{enumerate}
\end{remarks}

We next identify planar networks for which stable fixed points of Eq.~\eqref{eq:kuramoto_red} necessarily have all angle differences in $[-\pi/2,\pi/2]$, meaning that $\mathcal{N}^*=\mathcal{N}$. 
To that end, we recall a lemma by Taylor~\cite{Tay12} (a result very similar to this lemma can be found in Ref.~\onlinecite{Do12}).

\begin{lem}[Taylor~\cite{Tay12}]\label{lem:taylor}
 Let $\{\theta_i^{(0)}\}$ be a stable fixed point of Eq.~\eqref{eq:kuramoto_red} on a given graph $G$. 
 Then for any non-trivial partition $\mathcal{V}_G=U\cup U^c$ of the vertex set $\mathcal{V}_G$ of $G$,
 \begin{align*}
  \sum_{\substack{\langle ij\rangle\colon \\ i\in U,j\in U^c}}\cos(\theta_i^{(0)}-\theta_j^{(0)})\geq0\, .
 \end{align*}
\end{lem}

The following lemma characterizes network topologies which accept only stable fixed points with angle differences in $[-\pi/2,\pi/2]$ 
and for which the bound in Corollary~\ref{cor:main} applies for the total number $\mathcal{N}$ of stable fixed points. 
In networks where adjacent cycles share at least two consecutive edges, as in the left panel of Fig.~\ref{fig:twoEdges}, all stable fixed points have all angle differences in the interval 
$[-\pi/2,\pi/2]$. 
If however, two cycles share a single edge only, as in the right panel of Fig.~\ref{fig:oneEdge}, some stable fixed points may have angle differences larger than $\pi/2$.

\begin{lem}\label{lem:twoCommonEdges}
 On a graph $G$ where no pair of vertices with degree larger or equal to 3 are connected by a single edge, 
 all angle differences of any stable fixed point of Eq.~\eqref{eq:kuramoto_red} are in $[-\pi/2,\pi/2]$.
\end{lem}

\begin{proof}
 First, according to Lemma~\ref{lem:taylor}, the two edges connected to a vertex of degree 2 cannot both carry an angle difference whose cosines are negative. 
 Thus in the left panel of Fig.~\ref{fig:twoEdges}, $\Delta_{ji}$ and $\Delta_{ik}$ cannot be both larger than $\pi/2$. 
 Without loss of generality, we assume $\Delta_{ji}$ and $\Delta_{ik}$ to be both positive, otherwise we consider $\Delta_{ij}$ and $\Delta_{ki}$ instead. 
 \begin{figure}
  \begin{center}
   \includegraphics[width=.8\textwidth]{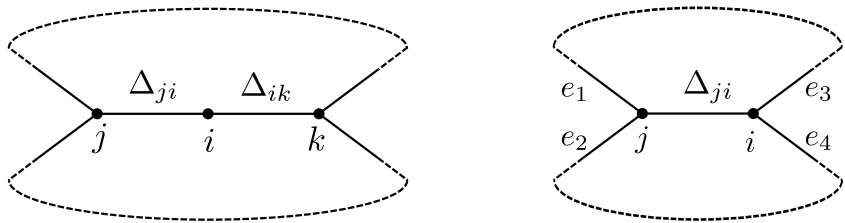}
   \caption{\it Neighboring cycles sharing two consecutive
   edges (left panel) and a single edge (right panel). In the left panel, 
   the angle differences $\Delta_{ji}$ and $\Delta_{ik}$ must both be in the interval $[-\pi/2,\pi/2]$
   for a fixed point solution to be stable. In the right panel, however, 
   the angle difference $\Delta_{ji}$ may be larger than $\pi/2$.}
   \label{fig:twoEdges}
   \label{fig:oneEdge}
  \end{center}
 \end{figure}
 Now if $\Delta_{ji}<\pi/2<\Delta_{ik}$, according to Eq.~\eqref{eq:kirchhoff}, 
 \begin{align*}
  \sin(-\Delta_{ji})+\sin(\Delta_{ik})=0& \iff \Delta_{ik}=\pi-\Delta_{ji}\\
  &\implies \cos(\Delta_{ik})=-\cos(\Delta_{ji})\eqqcolon -c\, .
 \end{align*}
 Consider the stability matrix $M$ defined in Eq.~\eqref{eq:stab_matrix}. The principal minor of $-M$ with row and column indices $\{j,i\}$ is 
 \begin{align}\label{eq:minor}
  \det
  \begin{pmatrix}
   x & -\cos(\Delta_{ji})\\
   -\cos(\Delta_{ji}) & \cos(\Delta_{ji})+\cos(\Delta_{ik})
  \end{pmatrix}
  &=\det
  \begin{pmatrix}
   x & -c \\
   -c & 0 
  \end{pmatrix}\, ,
 \end{align}
 which is negative. According to Sylvester's criterion~\cite{Hor86}, $M$ is not negative semidefinite, implying that the fixed point is unstable. 
 We conclude that a stable fixed point of Eq.~\eqref{eq:kuramoto_red} has all angle differences in $[-\pi/2,\pi/2]$. 
\end{proof}

Lemma~\ref{lem:twoCommonEdges} allows to apply Corollary~\ref{cor:main} to a well-defined class of networks. 
The next corollary gives a bound on the number of stable fixed points of Eq.~\eqref{eq:kuramoto_red} depending on the topology of the network. 

\begin{cor}\label{cor:topobound}
 On a  planar graph $G$ where no pair of vertices with degree larger or equal to 3 are connected by a single edge, 
 the number $\mathcal{N}$ of stable fixed points of Eq.~\eqref{eq:kuramoto_red} is bounded from
 above by 
 \begin{align}\label{eq:topobound}
  \mathcal{N}&\leq\prod_{k=1}^c\left[2\cdot{\rm Int}\left(n_k/4\right)+1\right]\, .
 \end{align}
\end{cor}

Corollary~\ref{cor:topobound} covers a large class of networks. 
The bound~\eqref{eq:topobound} is valid for any planar network with any number of cycles, as long as two cycles share either two or more edges, or no edge. 
When two cycles share a single edge, as in the right panel of Fig.~\ref{fig:oneEdge}, we cannot identify a negative principal minor of $-M$ as in Eq.~\eqref{eq:minor}, 
which invalidates the proof for such network topologies. 
In Sec.~\ref{sec:anglediff}, we discuss some examples where the stability matrix $M$ is negative semi-definite even though some angle differences are larger than $\pi/2$, 
meaning that all principal minors of $-M$ are non-negative. 
What happens in these situations, is that angle differences on edges $e_1$, $e_2$, $e_3$ and $e_4$ (see right panel of Fig.~\ref{fig:oneEdge}) may be small enough 
to stabilize an angle difference larger that $\pi/2$ on edge $\langle ij\rangle$.

\section{Angle differences exceeding \texorpdfstring{$\pi/2$}{}}\label{sec:anglediff}
The results of Sec.~\ref{sec:bound} do not give a bound on the number of stable fixed points of Eq.~\eqref{eq:kuramoto_red} for any planar network. 
There exist some examples of networks admitting a stable fixed point with some angle differences larger than $\pi/2$. 
We discuss two such examples here.
These networks have at least one pair of cycles sharing only one edge and consequently they do not satisfy the hypothesis of Lemma~\ref{lem:twoCommonEdges}. 

\paragraph*{Example 1}
Consider the network of Fig.~\ref{fig:2cyclesCounterex} with $K=1$, where the left path from $A$ to $B$ is of length $m_1=15$, 
the right path is of length $m_2=6$ and the center path is of lenght $m_{12}=1$. 
\begin{figure}
 \begin{center}
  \includegraphics[width=7cm]{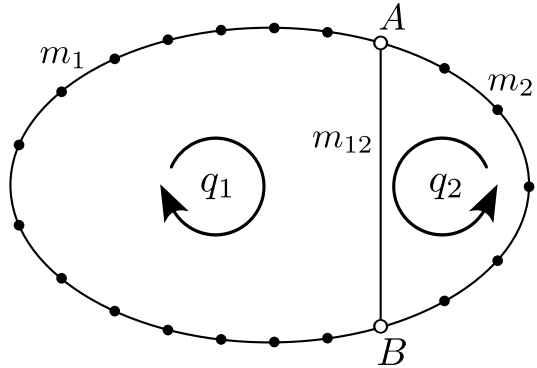}
  \caption{\it Example of a network having 
  a stable fixed point solution of Eq.~(\ref{eq:kuramoto_red}) with one angle 
  difference larger than $\pi/2$. 
  The number of edges along the three paths from $A$ to $B$ 
  are $m_1=15$, $m_2=6$ and $m_{12}=1$, the coupling constant is $K=1$. 
  The angle difference in the middle path is $\Delta_{AB}\approx\pi/2+0.149$ for a stable fixed
  point solution with $q_1=-1$ and $q_2=1$. }
  \label{fig:2cyclesCounterex}
 \end{center}
\end{figure}
We checked numerically that there exists a stable fixed point of Eq.~\eqref{eq:kuramoto_red} with one angle difference larger than $\pi/2$.
This fixed point has winding numbers $q_1=-1$ and $q_2=1$, a large angle difference on the central edge $\langle AB\rangle$, $\Delta_{AB}\approx\pi/2+0.149$ and a negative Lyapunov exponent,
$\lambda_2\approx-0.0194$. 

\paragraph*{Example 2}
Consider now the triangular lattice shown in the left panel of Fig.~\ref{fig:triangleCounterEx} with $K=1$. 
The fixed point with winding number $q=1$ on the central triangle has three angle differences of $2\pi/3$, 
but we checked numerically that it is stable nevertheless, $\lambda_2\approx-0.0974<0$. 
According to Refs.~\onlinecite{Del16,Tav72}, the fixed point with a vortex flow on the 3-vertex network (top right panel of Fig.~\ref{fig:triangleCounterEx}) is unstable. 
The largest eigenvalue of its stablility matrix can be analytically computed, $\lambda_2=1.5>0$. 
In the triangular lattice, the structure of the network surrounding the central triangle stabilizes these large angle differences, thanks to an increased connectivity. 
It should be noted however that increasing the connectivity with only tree-like graph extensions is not enough.  
The \emph{hairy triangle} shown on the bottom right panel of Fig.~\ref{fig:triangleCounterEx} has a dynamics that is the same as the 3-vertex network of the top right panel, 
$\lambda_2\approx0.386>0$. 
It therefore cannot carry a vortex flow. 
\begin{figure}
\begin{center}
 \includegraphics[width=260px]{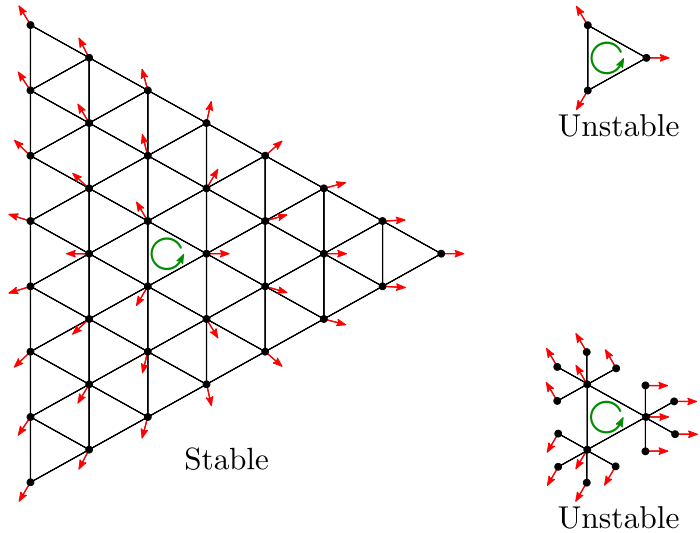}
 \caption{\it Left: Stable fixed point on a triangular lattice with three angle differences of $2\pi/3$, here $\lambda_2\approx-0.0974<0$. 
 The red arrows indicate the angle at each vertex and the green arrow indicates a vortex flow $q=1$ on the central triangle. The lattice shown is the smallest one we found numerically that stabilizes 
 this vortex flow.
 Top right: Fixed point on the triangular network with winding number $q=1$. 
 This fixed point is unstable, $\lambda_2=1.5>0$, which agrees with the results of Refs.~\onlinecite{Del16,Tav72}. 
 Bottom right: Fixed point on the \emph{hairy triangle} network with a vortex flow $q=1$. 
 Its dynamics is the same as the simple triangle, $\lambda_2\approx0.386$, even if its vertices have higher degree. 
 For all these three examples, $K=1$.}
 \label{fig:triangleCounterEx}
\end{center}
\end{figure}

\section{Conclusion}\label{sec:conclusion}

Our investigations allowed to give an upper bound, Eq.~\eqref{eq:topobound}, for the number 
of stable fixed point solutions with angle differences in $[-\pi/2,\pi/2]$, for the Kuramoto model with equal frequencies on planar networks. 
This bound is algebraic in the number of nodes on each cycle of the network, which is a
very significant improvement on earlier bounds 
which were exponential in the number of nodes, for all fixed point solutions 
(stable and unstable).~\cite{Meh15,Che16}

The structure of Eq.~\eqref{eq:topobound} identifies two opposite tendencies determining $\mathcal{N}$. 
On one hand more cycles means more terms in the product, which would suggest a larger number of stable fixed points. 
Furthermore, having more cycles increases the possiblity to have two cycles sharing a unique edge, which allows larger angle differences and possibly more stable fixed points. 
On the other hand, more cycles for a fixed number of vertices implies shorter cycles, and according to Eq.~\eqref{eq:topobound} again, this suggests less stable fixed points. 
The number of stable fixed points is then a trade-off between these two tendencies. 
This trade-off can be illustrated by considering the average number $\langle\mathcal{N}\rangle$ of stable fixed points of Eq.~\eqref{eq:kuramoto_red}, taken over all the networks with $n$ vertices 
and $c$ cycles. 
For $c=0$, all graph realizations are trees, which have a unique stable fixed point solution. 
For $c=1$, we already know that the average number of stable fixed points will increase and be bounded by $2\cdot{\rm Int}\left[n/4\right]+1$.~\cite{Del16} 
But, as $c$ increases further, 
the average number of stable fixed points will start to decrease at some point, 
to finally reach a unique stable fixed point for the all-to-all Kuramoto model.~\cite{Aey04,Mir05} 
We can then expect $\langle\mathcal{N}\rangle$ to increase for $c<c^*$ and decrease for $c>c^*$ with some critical value $c^*\in\left\{0,...,(n-1)(n-2)/2\right\}$. 
A non-monotonous behavior of the average $\langle\mathcal{N}\rangle$ with respect to $c$ has  
been reported numerically by Mehta et al.\cite{Meh15}

Equal frequencies correspond to the $K \rightarrow \infty$ limit of the Kuramoto model, Eq.~\eqref{eq:kuramoto_general}, and it is an open question whether our upper bound, 
Eq.~\eqref{eq:topobound}, remains
valid upon reducing $K$, when the spectrum of natural frequencies $P_i$ influences the 
fixed points. In Ref.~\onlinecite{Del16}, we managed to show that any stable fixed point
solution at a given value of $K$ remains stable upon increasing $K$ for a single-cycle 
network, meaning conversely that new stable fixed point solutions do not appear as $K$ is reduced. 
While we expect the same to hold for any network, we have been unable so far
to rigorously extend the argument of Ref.~\onlinecite{Del16} beyond single-cycle networks. 
We expect Eq.~\eqref{eq:topobound} to be an upper bound for any $K$, but have not been able to rigorously prove it. 

We showed the bound \eqref{eq:topobound} to be valid for planar networks. 
Our approach is valid for any network topology until Theorem~\ref{thm:main_thm}, which only applies to planar networks. 
All later results cannot be applied directly to general network topologies. 
Furthermore, our derivation of Eq.~\eqref{eq:topobound}, is valid under
the assumption that all stable fixed point solutions have angle differences 
$|\Delta_e | \le \pi/2$ on all edges $e$ of the network. In Section~\ref{sec:anglediff}
we identified topological ingredients that allow stable fixed point solutions
to violate that condition. We found that networks with
either no or strictly more than one edge common to any two cycles carry stable fixed point solutions
with only $|\Delta_e | \le \pi/2$, therefore Eq.~\eqref{eq:topobound} holds for the number of stable fixed points on such graphs. 
For graphs with edges single-shared by two cycles, the winding number on one such cycle is bounded from above by 
\begin{align*}
 |2\pi q_k|&\leq(n_k-n'_k)\pi/2+n'_k\pi
\end{align*}
with $n'_k$ counting the number of edges on the $k^{\rm th}$ cycle single-shared with another cycle 
(as in the right panel of Fig.~\ref{fig:oneEdge}) 
and for which an angle difference larger than $\pi/2$ does not necessarily lead to instability. 
We therefore conjecture the following upper bound for the number of stable fixed points on generic planar networks, 
\begin{align}\label{eq:conjecture}
 \mathcal{N}&\leq\prod_{k=1}^c\left[2\cdot{\rm Int}\left((n_k+n'_k)/4\right)+1\right]\, .
\end{align}
In trying to rigorously confirm this bound we could not prove
the one-to-one correspondence between fixed point solutions and vortex flow vectors when $|\Delta_e | > \pi/2$ for some edges.
Eq.~\eqref{eq:conjecture} is therefore a conjecture which we could not disprove numerically, but have yet to prove rigorously.

We finally point out that the approach based on vortex flows we constructed is equivalent,
in quantum-mechanical language, to rewriting the problem in a basis with good 
angular momentum quantum number $\{\ell_k\}$ for each fundamental cycle. Given the knowledge of 
all admissible values for each $\ell_k$ it is then possible to construct initial states with higher chances
to converge towards fixed point solutions with finite winding vector $\vec{q}$. A numerical algorithm to determine 
all possible stable fixed point solutions in complex planar networks has been constructed and will be 
presented elsewhere.

After completion of this manuscript, we heard of Ref.~\onlinecite{Man16}, which reproduces several of our results.

\section{Acknowledgment}

We thank J. Dubuisson and L. Pagnier for interesting discussions. This work has been 
supported by the Swiss National Science
Foundation under an AP Energy Grant.

\bibliographystyle{apsrev4-1}
\bibliography{bibliography}

\begin{thebibliography}{35}%
\makeatletter
\providecommand \@ifxundefined [1]{%
 \@ifx{#1\undefined}
}%
\providecommand \@ifnum [1]{%
 \ifnum #1\expandafter \@firstoftwo
 \else \expandafter \@secondoftwo
 \fi
}%
\providecommand \@ifx [1]{%
 \ifx #1\expandafter \@firstoftwo
 \else \expandafter \@secondoftwo
 \fi
}%
\providecommand \natexlab [1]{#1}%
\providecommand \enquote  [1]{``#1''}%
\providecommand \bibnamefont  [1]{#1}%
\providecommand \bibfnamefont [1]{#1}%
\providecommand \citenamefont [1]{#1}%
\providecommand \href@noop [0]{\@secondoftwo}%
\providecommand \href [0]{\begingroup \@sanitize@url \@href}%
\providecommand \@href[1]{\@@startlink{#1}\@@href}%
\providecommand \@@href[1]{\endgroup#1\@@endlink}%
\providecommand \@sanitize@url [0]{\catcode `\\12\catcode `\$12\catcode
  `\&12\catcode `\#12\catcode `\^12\catcode `\_12\catcode `\%12\relax}%
\providecommand \@@startlink[1]{}%
\providecommand \@@endlink[0]{}%
\providecommand \url  [0]{\begingroup\@sanitize@url \@url }%
\providecommand \@url [1]{\endgroup\@href {#1}{\urlprefix }}%
\providecommand \urlprefix  [0]{URL }%
\providecommand \Eprint [0]{\href }%
\providecommand \doibase [0]{http://dx.doi.org/}%
\providecommand \selectlanguage [0]{\@gobble}%
\providecommand \bibinfo  [0]{\@secondoftwo}%
\providecommand \bibfield  [0]{\@secondoftwo}%
\providecommand \translation [1]{[#1]}%
\providecommand \BibitemOpen [0]{}%
\providecommand \bibitemStop [0]{}%
\providecommand \bibitemNoStop [0]{.\EOS\space}%
\providecommand \EOS [0]{\spacefactor3000\relax}%
\providecommand \BibitemShut  [1]{\csname bibitem#1\endcsname}%
\let\auto@bib@innerbib\@empty
\bibitem [{\citenamefont {Kuramoto}(1975)}]{Kur75abbrv}%
  \BibitemOpen
  \bibfield  {author} {\bibinfo {author} {\bibfnamefont {Y.}~\bibnamefont
  {Kuramoto}},\ }in\ \href {\doibase 10.1007/BFb0013365} {\emph {\bibinfo
  {booktitle} {Lecture Notes in Physics, vol. 39.}}},\ \bibinfo {editor}
  {edited by\ \bibinfo {editor} {\bibfnamefont {H.}~\bibnamefont {Araki}}}\
  (\bibinfo  {publisher} {Springer},\ \bibinfo {address} {Berlin, Heidelberg},\
  \bibinfo {year} {1975})\ pp.\ \bibinfo {pages} {420--422}\BibitemShut
  {NoStop}%
\bibitem [{\citenamefont {Kuramoto}(1984)}]{Kur84}%
  \BibitemOpen
  \bibfield  {author} {\bibinfo {author} {\bibfnamefont {Y.}~\bibnamefont
  {Kuramoto}},\ }\href@noop {} {\bibfield  {journal} {\bibinfo  {journal}
  {Progr. Theoret. Phys. Suppl.}\ }\textbf {\bibinfo {volume} {79}},\ \bibinfo
  {pages} {223} (\bibinfo {year} {1984})}\BibitemShut {NoStop}%
\bibitem [{\citenamefont {Strogatz}(2000)}]{Str00}%
  \BibitemOpen
  \bibfield  {author} {\bibinfo {author} {\bibfnamefont {S.~H.}\ \bibnamefont
  {Strogatz}},\ }\href@noop {} {\bibfield  {journal} {\bibinfo  {journal}
  {Physica D}\ }\textbf {\bibinfo {volume} {143}},\ \bibinfo {pages} {1}
  (\bibinfo {year} {2000})}\BibitemShut {NoStop}%
\bibitem [{\citenamefont {Acebr\'on}\ \emph {et~al.}(2005)\citenamefont
  {Acebr\'on}, \citenamefont {Bonilla}, \citenamefont {P\'erez~Vicente},
  \citenamefont {Ritort},\ and\ \citenamefont {Spigler}}]{Ace05}%
  \BibitemOpen
  \bibfield  {author} {\bibinfo {author} {\bibfnamefont {J.~A.}\ \bibnamefont
  {Acebr\'on}}, \bibinfo {author} {\bibfnamefont {L.~L.}\ \bibnamefont
  {Bonilla}}, \bibinfo {author} {\bibfnamefont {C.~J.}\ \bibnamefont
  {P\'erez~Vicente}}, \bibinfo {author} {\bibfnamefont {F.}~\bibnamefont
  {Ritort}}, \ and\ \bibinfo {author} {\bibfnamefont {R.}~\bibnamefont
  {Spigler}},\ }\href@noop {} {\bibfield  {journal} {\bibinfo  {journal} {Rev.
  Mod. Phys.}\ }\textbf {\bibinfo {volume} {77}},\ \bibinfo {pages} {137}
  (\bibinfo {year} {2005})}\BibitemShut {NoStop}%
\bibitem [{\citenamefont {D\"orfler}\ and\ \citenamefont
  {Bullo}(2014)}]{Dor14}%
  \BibitemOpen
  \bibfield  {author} {\bibinfo {author} {\bibfnamefont {F.}~\bibnamefont
  {D\"orfler}}\ and\ \bibinfo {author} {\bibfnamefont {F.}~\bibnamefont
  {Bullo}},\ }\href@noop {} {\bibfield  {journal} {\bibinfo  {journal}
  {Automatica}\ }\textbf {\bibinfo {volume} {50}},\ \bibinfo {pages} {1539}
  (\bibinfo {year} {2014})}\BibitemShut {NoStop}%
\bibitem [{\citenamefont {Ermentrout}(1985)}]{Erm85}%
  \BibitemOpen
  \bibfield  {author} {\bibinfo {author} {\bibfnamefont {G.~B.}\ \bibnamefont
  {Ermentrout}},\ }\href@noop {} {\bibfield  {journal} {\bibinfo  {journal} {J.
  Math. Biol.}\ }\textbf {\bibinfo {volume} {22}},\ \bibinfo {pages} {1}
  (\bibinfo {year} {1985})}\BibitemShut {NoStop}%
\bibitem [{\citenamefont {van Hemmen}\ and\ \citenamefont
  {Wreskinski}(1993)}]{Hem93}%
  \BibitemOpen
  \bibfield  {author} {\bibinfo {author} {\bibfnamefont {J.~L.}\ \bibnamefont
  {van Hemmen}}\ and\ \bibinfo {author} {\bibfnamefont {W.~F.}\ \bibnamefont
  {Wreskinski}},\ }\href@noop {} {\bibfield  {journal} {\bibinfo  {journal} {J.
  Stat. Phys.}\ }\textbf {\bibinfo {volume} {72}},\ \bibinfo {pages} {145}
  (\bibinfo {year} {1993})}\BibitemShut {NoStop}%
\bibitem [{\citenamefont {Aeyels}\ and\ \citenamefont {Rogge}(2004)}]{Aey04}%
  \BibitemOpen
  \bibfield  {author} {\bibinfo {author} {\bibfnamefont {D.}~\bibnamefont
  {Aeyels}}\ and\ \bibinfo {author} {\bibfnamefont {J.~A.}\ \bibnamefont
  {Rogge}},\ }\href@noop {} {\bibfield  {journal} {\bibinfo  {journal} {Prog.
  Th. Phys.}\ }\textbf {\bibinfo {volume} {112}},\ \bibinfo {pages} {921}
  (\bibinfo {year} {2004})}\BibitemShut {NoStop}%
\bibitem [{\citenamefont {Mirollo}\ and\ \citenamefont
  {Strogatz}(2005)}]{Mir05}%
  \BibitemOpen
  \bibfield  {author} {\bibinfo {author} {\bibfnamefont {R.~E.}\ \bibnamefont
  {Mirollo}}\ and\ \bibinfo {author} {\bibfnamefont {S.~H.}\ \bibnamefont
  {Strogatz}},\ }\href@noop {} {\bibfield  {journal} {\bibinfo  {journal}
  {Physica D}\ }\textbf {\bibinfo {volume} {205}},\ \bibinfo {pages} {249}
  (\bibinfo {year} {2005})}\BibitemShut {NoStop}%
\bibitem [{\citenamefont {Mehta}\ \emph {et~al.}(2015)\citenamefont {Mehta},
  \citenamefont {Daleo}, \citenamefont {D\"orfler},\ and\ \citenamefont
  {Hauenstein}}]{Meh15}%
  \BibitemOpen
  \bibfield  {author} {\bibinfo {author} {\bibfnamefont {D.}~\bibnamefont
  {Mehta}}, \bibinfo {author} {\bibfnamefont {N.~S.}\ \bibnamefont {Daleo}},
  \bibinfo {author} {\bibfnamefont {F.}~\bibnamefont {D\"orfler}}, \ and\
  \bibinfo {author} {\bibfnamefont {J.~D.}\ \bibnamefont {Hauenstein}},\ }\href
  {\doibase 10.1063/1.4919696} {\bibfield  {journal} {\bibinfo  {journal}
  {Chaos}\ }\textbf {\bibinfo {volume} {25}},\ \bibinfo {pages} {053103}
  (\bibinfo {year} {2015})}\BibitemShut {NoStop}%
\bibitem [{\citenamefont {Chen}\ \emph {et~al.}(2016)\citenamefont {Chen},
  \citenamefont {Mehta},\ and\ \citenamefont {Niemerg}}]{Che16}%
  \BibitemOpen
  \bibfield  {author} {\bibinfo {author} {\bibfnamefont {T.}~\bibnamefont
  {Chen}}, \bibinfo {author} {\bibfnamefont {D.}~\bibnamefont {Mehta}}, \ and\
  \bibinfo {author} {\bibfnamefont {M.}~\bibnamefont {Niemerg}},\ }\href
  {http://arxiv.org/abs/1603.05905} {\bibfield  {journal} {\bibinfo  {journal}
  {arXiv:1603.05905}\ } (\bibinfo {year} {2016})}\BibitemShut {NoStop}%
\bibitem [{\citenamefont {Rogge}\ and\ \citenamefont {Aeyels}(2004)}]{Rog04}%
  \BibitemOpen
  \bibfield  {author} {\bibinfo {author} {\bibfnamefont {J.~A.}\ \bibnamefont
  {Rogge}}\ and\ \bibinfo {author} {\bibfnamefont {D.}~\bibnamefont {Aeyels}},\
  }\href@noop {} {\bibfield  {journal} {\bibinfo  {journal} {J. Phys. A}\
  }\textbf {\bibinfo {volume} {37}},\ \bibinfo {pages} {11135} (\bibinfo {year}
  {2004})}\BibitemShut {NoStop}%
\bibitem [{\citenamefont {Wiley}\ \emph {et~al.}(2006)\citenamefont {Wiley},
  \citenamefont {Strogatz},\ and\ \citenamefont {Girvan}}]{Wil06}%
  \BibitemOpen
  \bibfield  {author} {\bibinfo {author} {\bibfnamefont {D.~A.}\ \bibnamefont
  {Wiley}}, \bibinfo {author} {\bibfnamefont {S.~H.}\ \bibnamefont {Strogatz}},
  \ and\ \bibinfo {author} {\bibfnamefont {M.}~\bibnamefont {Girvan}},\ }\href
  {\doibase 10.1063/1.2165594} {\bibfield  {journal} {\bibinfo  {journal}
  {Chaos}\ }\textbf {\bibinfo {volume} {16}},\ \bibinfo {pages} {015103}
  (\bibinfo {year} {2006})}\BibitemShut {NoStop}%
\bibitem [{\citenamefont {Menck}\ \emph {et~al.}(2014)\citenamefont {Menck},
  \citenamefont {Heitzig}, \citenamefont {Kurths},\ and\ \citenamefont
  {Schellnhuber}}]{Men14}%
  \BibitemOpen
  \bibfield  {author} {\bibinfo {author} {\bibfnamefont {P.~J.}\ \bibnamefont
  {Menck}}, \bibinfo {author} {\bibfnamefont {J.}~\bibnamefont {Heitzig}},
  \bibinfo {author} {\bibfnamefont {J.}~\bibnamefont {Kurths}}, \ and\ \bibinfo
  {author} {\bibfnamefont {H.~J.}\ \bibnamefont {Schellnhuber}},\ }\href@noop
  {} {\bibfield  {journal} {\bibinfo  {journal} {Nat. Comms.}\ }\textbf
  {\bibinfo {volume} {5}},\ \bibinfo {pages} {3969} (\bibinfo {year}
  {2014})}\BibitemShut {NoStop}%
\bibitem [{\citenamefont {D\"orfler}\ \emph {et~al.}(2013)\citenamefont
  {D\"orfler}, \citenamefont {Chertkov},\ and\ \citenamefont {Bullo}}]{Dor13}%
  \BibitemOpen
  \bibfield  {author} {\bibinfo {author} {\bibfnamefont {F.}~\bibnamefont
  {D\"orfler}}, \bibinfo {author} {\bibfnamefont {M.}~\bibnamefont {Chertkov}},
  \ and\ \bibinfo {author} {\bibfnamefont {F.}~\bibnamefont {Bullo}},\
  }\href@noop {} {\bibfield  {journal} {\bibinfo  {journal} {Proc. Natl. Acad.
  Sci.}\ }\textbf {\bibinfo {volume} {110}},\ \bibinfo {pages} {2005} (\bibinfo
  {year} {2013})}\BibitemShut {NoStop}%
\bibitem [{\citenamefont {Delabays}\ \emph {et~al.}(2016)\citenamefont
  {Delabays}, \citenamefont {Coletta},\ and\ \citenamefont {Jacquod}}]{Del16}%
  \BibitemOpen
  \bibfield  {author} {\bibinfo {author} {\bibfnamefont {R.}~\bibnamefont
  {Delabays}}, \bibinfo {author} {\bibfnamefont {T.}~\bibnamefont {Coletta}}, \
  and\ \bibinfo {author} {\bibfnamefont {P.}~\bibnamefont {Jacquod}},\
  }\href@noop {} {\bibfield  {journal} {\bibinfo  {journal} {J. Math. Phys.}\
  }\textbf {\bibinfo {volume} {57}},\ \bibinfo {pages} {032701} (\bibinfo
  {year} {2016})}\BibitemShut {NoStop}%
\bibitem [{\citenamefont {Korsak}(1972)}]{Kor72}%
  \BibitemOpen
  \bibfield  {author} {\bibinfo {author} {\bibfnamefont {A.~J.}\ \bibnamefont
  {Korsak}},\ }\href
  {http://ieeexplore.ieee.org/xpls/abs_all.jsp?arnumber=4074824} {\bibfield
  {journal} {\bibinfo  {journal} {IEEE Trans. Power App. Syst.}\ }\textbf
  {\bibinfo {volume} {PAS-91}},\ \bibinfo {pages} {1093} (\bibinfo {year}
  {1972})}\BibitemShut {NoStop}%
\bibitem [{\citenamefont {Tavora}\ and\ \citenamefont {Smith}(1972)}]{Tav72}%
  \BibitemOpen
  \bibfield  {author} {\bibinfo {author} {\bibfnamefont {C.~J.}\ \bibnamefont
  {Tavora}}\ and\ \bibinfo {author} {\bibfnamefont {O.~J.~M.}\ \bibnamefont
  {Smith}},\ }\href
  {http://ieeexplore.ieee.org/xpls/abs_all.jsp?arnumber=4074831} {\bibfield
  {journal} {\bibinfo  {journal} {IEEE Trans. Power App. Syst.}\ }\textbf
  {\bibinfo {volume} {PAS-91}},\ \bibinfo {pages} {1138} (\bibinfo {year}
  {1972})}\BibitemShut {NoStop}%
\bibitem [{\citenamefont {Bergen}\ and\ \citenamefont {Vittal}(2000)}]{Ber00}%
  \BibitemOpen
  \bibfield  {author} {\bibinfo {author} {\bibfnamefont {A.~R.}\ \bibnamefont
  {Bergen}}\ and\ \bibinfo {author} {\bibfnamefont {V.}~\bibnamefont
  {Vittal}},\ }\href@noop {} {\emph {\bibinfo {title} {Power {Systems}
  {Analysis}}}}\ (\bibinfo  {publisher} {Prentice Hall},\ \bibinfo {year}
  {2000})\BibitemShut {NoStop}%
\bibitem [{\citenamefont {Skar}(1980)}]{Ska80}%
  \BibitemOpen
  \bibfield  {author} {\bibinfo {author} {\bibfnamefont {S.~J.}\ \bibnamefont
  {Skar}},\ }\emph {\bibinfo {title} {{Stability} of {Power} {Systems} and
  other {Systems} of {Second} {Order} {Differential} {Equations}}},\ \href@noop
  {} {Ph.D. thesis},\ \bibinfo  {school} {Iowa State University} (\bibinfo
  {year} {1980})\BibitemShut {NoStop}%
\bibitem [{\citenamefont {Araposthatis}\ \emph {et~al.}(1981)\citenamefont
  {Araposthatis}, \citenamefont {Sastry},\ and\ \citenamefont
  {Varayia}}]{Ara81}%
  \BibitemOpen
  \bibfield  {author} {\bibinfo {author} {\bibfnamefont {A.}~\bibnamefont
  {Araposthatis}}, \bibinfo {author} {\bibfnamefont {S.}~\bibnamefont
  {Sastry}}, \ and\ \bibinfo {author} {\bibfnamefont {P.}~\bibnamefont
  {Varayia}},\ }\href@noop {} {\bibfield  {journal} {\bibinfo  {journal} {Int.
  J. Elect. Power Energy Syst.}\ }\textbf {\bibinfo {volume} {3}},\ \bibinfo
  {pages} {115} (\bibinfo {year} {1981})}\BibitemShut {NoStop}%
\bibitem [{\citenamefont {Bailleul}\ and\ \citenamefont
  {Byrnes}(1982)}]{Bai82}%
  \BibitemOpen
  \bibfield  {author} {\bibinfo {author} {\bibfnamefont {J.}~\bibnamefont
  {Bailleul}}\ and\ \bibinfo {author} {\bibfnamefont {C.~I.}\ \bibnamefont
  {Byrnes}},\ }\href@noop {} {\bibfield  {journal} {\bibinfo  {journal} {Proc.
  of the 21$^{\rm st}$ IEEE Conf. on Decision and Control}\ }\textbf {\bibinfo
  {volume} {2}},\ \bibinfo {pages} {919} (\bibinfo {year} {1982})}\BibitemShut
  {NoStop}%
\bibitem [{\citenamefont {Tamura}\ \emph {et~al.}(1983)\citenamefont {Tamura},
  \citenamefont {Mori},\ and\ \citenamefont {Iwamoto}}]{Tam83}%
  \BibitemOpen
  \bibfield  {author} {\bibinfo {author} {\bibfnamefont {Y.}~\bibnamefont
  {Tamura}}, \bibinfo {author} {\bibfnamefont {H.}~\bibnamefont {Mori}}, \ and\
  \bibinfo {author} {\bibfnamefont {S.}~\bibnamefont {Iwamoto}},\ }\href
  {\doibase 10.1109/TPAS.1983.318052} {\bibfield  {journal} {\bibinfo
  {journal} {IEEE Trans. Power App. Syst.}\ }\textbf {\bibinfo {volume}
  {PAS-102}},\ \bibinfo {pages} {1115} (\bibinfo {year} {1983})}\BibitemShut
  {NoStop}%
\bibitem [{\citenamefont {Klos}\ and\ \citenamefont {Wojcicka}(1991)}]{Klos91}%
  \BibitemOpen
  \bibfield  {author} {\bibinfo {author} {\bibfnamefont {A.}~\bibnamefont
  {Klos}}\ and\ \bibinfo {author} {\bibfnamefont {J.}~\bibnamefont
  {Wojcicka}},\ }\href {\doibase
  http://dx.doi.org/10.1016/0142-0615(91)90050-6} {\bibfield  {journal}
  {\bibinfo  {journal} {Int. J. Elect. Power Energy Syst.}\ }\textbf {\bibinfo
  {volume} {13}},\ \bibinfo {pages} {268 } (\bibinfo {year}
  {1991})}\BibitemShut {NoStop}%
\bibitem [{\citenamefont {Nguyen}\ and\ \citenamefont
  {Turitsyn}(2014)}]{Ngu14}%
  \BibitemOpen
  \bibfield  {author} {\bibinfo {author} {\bibfnamefont {H.~D.}\ \bibnamefont
  {Nguyen}}\ and\ \bibinfo {author} {\bibfnamefont {K.~S.}\ \bibnamefont
  {Turitsyn}},\ }in\ \href {\doibase 10.1109/PESGM.2014.6938797} {\emph
  {\bibinfo {booktitle} {PES General Meeting - Conference Exposition, 2014
  IEEE}}}\ (\bibinfo {year} {2014})\BibitemShut {NoStop}%
\bibitem [{\citenamefont {Janssens}\ and\ \citenamefont
  {Kamagate}(2003)}]{Jan03}%
  \BibitemOpen
  \bibfield  {author} {\bibinfo {author} {\bibfnamefont {N.}~\bibnamefont
  {Janssens}}\ and\ \bibinfo {author} {\bibfnamefont {A.}~\bibnamefont
  {Kamagate}},\ }\href@noop {} {\bibfield  {journal} {\bibinfo  {journal} {Int.
  J. Elect. Power Energy Syst.}\ }\textbf {\bibinfo {volume} {25}},\ \bibinfo
  {pages} {591 } (\bibinfo {year} {2003})}\BibitemShut {NoStop}%
\bibitem [{\citenamefont {Coletta}\ \emph {et~al.}(2016)\citenamefont
  {Coletta}, \citenamefont {Delabays}, \citenamefont {Adagideli},\ and\
  \citenamefont {Jacquod}}]{Col16b}%
  \BibitemOpen
  \bibfield  {author} {\bibinfo {author} {\bibfnamefont {T.}~\bibnamefont
  {Coletta}}, \bibinfo {author} {\bibfnamefont {R.}~\bibnamefont {Delabays}},
  \bibinfo {author} {\bibfnamefont {I.}~\bibnamefont {Adagideli}}, \ and\
  \bibinfo {author} {\bibfnamefont {P.}~\bibnamefont {Jacquod}},\ }\href@noop
  {} {\bibfield  {journal} {\bibinfo  {journal} {New J. Phys.}\ }\textbf
  {\bibinfo {volume} {18}} (\bibinfo {year} {2016})}\BibitemShut {NoStop}%
\bibitem [{\citenamefont {Ochab}\ and\ \citenamefont {G\'ora}(2010)}]{Och10}%
  \BibitemOpen
  \bibfield  {author} {\bibinfo {author} {\bibfnamefont {J.}~\bibnamefont
  {Ochab}}\ and\ \bibinfo {author} {\bibfnamefont {P.~F.}\ \bibnamefont
  {G\'ora}},\ }\href@noop {} {\bibfield  {journal} {\bibinfo  {journal} {Acta
  Phys. Pol. B [Proc. Suppl. 3]}\ }\textbf {\bibinfo {volume} {3}},\ \bibinfo
  {pages} {453} (\bibinfo {year} {2010})}\BibitemShut {NoStop}%
\bibitem [{\citenamefont {Tilles}\ \emph {et~al.}(2011)\citenamefont {Tilles},
  \citenamefont {Ferreira},\ and\ \citenamefont {Cerdeira}}]{Til11}%
  \BibitemOpen
  \bibfield  {author} {\bibinfo {author} {\bibfnamefont {P.~F.~C.}\
  \bibnamefont {Tilles}}, \bibinfo {author} {\bibfnamefont {F.~F.}\
  \bibnamefont {Ferreira}}, \ and\ \bibinfo {author} {\bibfnamefont {H.~A.}\
  \bibnamefont {Cerdeira}},\ }\href
  {http://link.aps.org/doi/10.1103/PhysRevE.83.066206} {\bibfield  {journal}
  {\bibinfo  {journal} {Phys. Rev. E}\ }\textbf {\bibinfo {volume} {83}}
  (\bibinfo {year} {2011})}\BibitemShut {NoStop}%
\bibitem [{\citenamefont {Roy}\ and\ \citenamefont {Lahiri}(2012)}]{Roy12}%
  \BibitemOpen
  \bibfield  {author} {\bibinfo {author} {\bibfnamefont {T.~K.}\ \bibnamefont
  {Roy}}\ and\ \bibinfo {author} {\bibfnamefont {A.}~\bibnamefont {Lahiri}},\
  }\href {\doibase 10.1016/j.chaos.2012.03.004} {\bibfield  {journal} {\bibinfo
   {journal} {Chaos, Solitons \& Fractals}\ }\textbf {\bibinfo {volume} {45}},\
  \bibinfo {pages} {888} (\bibinfo {year} {2012})}\BibitemShut {NoStop}%
\bibitem [{\citenamefont {Taylor}(2012)}]{Tay12}%
  \BibitemOpen
  \bibfield  {author} {\bibinfo {author} {\bibfnamefont {R.}~\bibnamefont
  {Taylor}},\ }\href@noop {} {\bibfield  {journal} {\bibinfo  {journal} {J.
  Phys. A}\ }\textbf {\bibinfo {volume} {45}},\ \bibinfo {pages} {055102}
  (\bibinfo {year} {2012})}\BibitemShut {NoStop}%
\bibitem [{\citenamefont {Do}\ \emph {et~al.}(2012)\citenamefont {Do},
  \citenamefont {Boccaletti},\ and\ \citenamefont {Gross}}]{Do12}%
  \BibitemOpen
  \bibfield  {author} {\bibinfo {author} {\bibfnamefont {A.-L.}\ \bibnamefont
  {Do}}, \bibinfo {author} {\bibfnamefont {S.}~\bibnamefont {Boccaletti}}, \
  and\ \bibinfo {author} {\bibfnamefont {T.}~\bibnamefont {Gross}},\
  }\href@noop {} {\bibfield  {journal} {\bibinfo  {journal} {Phys. Rev. Lett.}\
  }\textbf {\bibinfo {volume} {108}},\ \bibinfo {pages} {194102} (\bibinfo
  {year} {2012})}\BibitemShut {NoStop}%
\bibitem [{\citenamefont {Biggs}(1993)}]{Big93}%
  \BibitemOpen
  \bibfield  {author} {\bibinfo {author} {\bibfnamefont {N.}~\bibnamefont
  {Biggs}},\ }\href@noop {} {\emph {\bibinfo {title} {Algebraic graph
  theory}}},\ \bibinfo {edition} {2nd}\ ed.\ (\bibinfo  {publisher} {Cambridge
  University Press},\ \bibinfo {year} {1993})\BibitemShut {NoStop}%
\bibitem [{\citenamefont {Horn}\ and\ \citenamefont {Johnson}(1986)}]{Hor86}%
  \BibitemOpen
  \bibfield  {author} {\bibinfo {author} {\bibfnamefont {R.~A.}\ \bibnamefont
  {Horn}}\ and\ \bibinfo {author} {\bibfnamefont {C.~R.}\ \bibnamefont
  {Johnson}},\ }\href@noop {} {\emph {\bibinfo {title} {Matrix Analysis}}}\
  (\bibinfo  {publisher} {Cambridge University Press},\ \bibinfo {address} {New
  York},\ \bibinfo {year} {1986})\BibitemShut {NoStop}%
\bibitem [{\citenamefont {Manik}\ \emph {et~al.}(2016)\citenamefont {Manik},
  \citenamefont {Timme},\ and\ \citenamefont {Witthaut}}]{Man16}%
  \BibitemOpen
  \bibfield  {author} {\bibinfo {author} {\bibfnamefont {D.}~\bibnamefont
  {Manik}}, \bibinfo {author} {\bibfnamefont {M.}~\bibnamefont {Timme}}, \ and\
  \bibinfo {author} {\bibfnamefont {D.}~\bibnamefont {Witthaut}},\ }\href
  {http://arxiv.org/abs/1611.09825} {\bibfield  {journal} {\bibinfo  {journal}
  {{arXiv}:1611.09825}\ } (\bibinfo {year} {2016})}\BibitemShut {NoStop}%
\end{thebibliography}%

\end{document}